\documentclass{article}
\usepackage{graphicx}
\usepackage{mathtools,amssymb,amsmath,amsthm}
\usepackage{algorithmicx,algorithm}
\usepackage[noend]{algpseudocode}
\usepackage{cleveref}
\usepackage{multirow}
\usepackage{subcaption}

\newtheorem{definition}{Definition}
\newtheorem{lemma}{Lemma}
\newtheorem{theorem}{Theorem}
\newcommand{\size}[1]{\left\lvert#1\right\rvert}

\DeclareMathOperator{\dist}{d}

\DeclareMathOperator{\polylog}{polylog}

\begin{document}
\title{A $2$-Approximation Algorithm for Data-Distributed Metric $k$-Center}
\author{Sepideh Aghamolaei and Mohammad Ghodsi\thanks{Department of Computer Engineering, Sharif University of Technology,
Tehran, Iran.}}

\maketitle
\begin{abstract}
In a metric space, a set of point sets of roughly the same size and an integer $k\geq 1$ are given as the input and the goal of data-distributed $k$-center is to find a subset of size $k$ of the input points as the set of centers to minimize the maximum distance from the input points to their closest centers. Metric $k$-center is known to be NP-hard which carries to the data-distributed setting.

We give a $2$-approximation algorithm of $k$-center for sublinear $k$ in the data-distributed setting, which is tight. This algorithm works in several models, including the massively parallel computation model (MPC).
\end{abstract}

\section{Introduction}
Data-distributed models are theoretical tools for designing algorithms for data center and cloud environments. In these settings, medium-sized chunks of data are stored in different servers that work in parallel rounds and can communicate over a relatively fast network. One of these models is the massively parallel computation model (MPC)~\cite{beame2013communication}. A simpler model is composable coresets~\cite{indyk2014composable} which is closer to the notion of coresets from computational geometry and which is a powerful enough tool to get approximation algorithms for some problems.
We discuss approximation algorithms for a clustering problem called metric $k$-center in such models, as the problem is $\mathbf{NP}$-hard to solve exactly in polynomial time. For more details on the definitions and the literature, see \Cref{sec:pre}.
The ultimate goal is to achieve the best possible approximation for metric $k$-center using a constant number of MPC rounds.

For metric $k$-center in MPC, a $4$-approximation algorithm based on Gonzalez's algorithm~\cite{gonzalez} with $2$ rounds exists that first computes a composable coreset and then finds a solution on the result~\cite{kcenter1}\footnote{The paper does not directly mention composable coresets but the idea is the same.}.
This is twice the approximation ratio of the sequential algorithm used as its subroutine.
A MPC algorithm with a logarithmic number of rounds and approximation ratio $2+\epsilon$, for any $\epsilon>0$ exists that uses geometrical guessing based on the value of the diameter and then random sampling to find a set of candidate centers~\cite{im2015fast}. Using the $4$-approximation instead of the diameter improves the round complexity to a constant amount that depends on $\epsilon$. A recent result claims the same approximation ratio~\cite{zarrabi}.

For a parallel algorithm, the total time complexities in the processors which is also equal to the time complexity of the sequential execution of that algorithm is called the work of the algorithm.
Because of the constraints of the model, which are reviewed in \Cref{sec:pre}, all MPC algorithms have near-linear work.
For large enough values of $k$, the time complexity of Gonzalez's algorithm which is $O(nk)$ becomes near quadratic.
An $O(\log \log \log n)$-approximation algorithm with $k(1+o(1))$ centers, $O(\log \log n)$ rounds and $O(n^{1+\epsilon}\polylog n)$ work has been presented~\cite{bateni21} that improves this bound on the work of the algorithm for $k=\Omega(n^{1-\epsilon})$, for a constant $\epsilon$. We focus on cases where the output fits inside the memory of a single machine in MPC, that is for values $k$ that are sublinear in $n$. This allows us to compare the approximation ratio of our algorithm with existing sequential algorithms.

Our approach is to add indices to input points and communicate partial solutions to all machines to allow each set to recover the same solution as other sets locally and independently. Previous works on this problem compute the partial solutions and work on the resulting subset of input points which results in a worse approximation ratio for the cost of clustering.

In \Cref{sec:defs}, we discuss definitions and previous work on the problem. In \Cref{sec:contrib}, we discuss our results and compare them with the state of the art on this problem. We denote the size of the input with $n$.

\subsection{Contributions}\label{sec:contrib}
A rough sketch of our main result is a two-phase algorithm that first computes a composable coreset for $k$-center and then uses this coreset to solve $k$-center by coordinating the machines using a fixed ordering over the coreset from the first phase.
This improves the approximation factor of $k$-center in MPC to $2$.

A comparison of the results for $k$-center is given in \Cref{table:results}.
\begin{table}[h]
\centering
\begin{tabular}{|l|l|l|}
\hline
Model & Approximation ratio & Size\\
\hline
\hline
 \multirow{2}{*}{{MPC}}& $4$~\cite{kcenter1} & $O(kL)$\\
 & $2+\epsilon$~\cite{im2015fast,kcenter1} (randomized) & $O(k\sqrt{n}\log n)$ \\
 & $2+\epsilon$~\cite{zarrabi} & $O(kn\polylog (n))$ \\
 & $2$ & $O(k^2L)$\\
\hline
any polynomial-time algorithm & $\geq2$~\cite{vazirani2013approximation} & -\\
\hline
\end{tabular}
\caption{A summary of results on approximation algorithms for metric $k$-center in MPC assuming $k$ is sublinear. Here, $\epsilon$ is an arbitrary positive constant.}
\label{table:results}
\end{table}

Throughout the paper, we assume $k=o(n)$, because if $k=n$, then, it is enough to simply return all the input points as the output. If $k=\Theta(n)$, then it is relatively meaningless to find a coreset for the problem. Also, $k$-center means metric $k$-center, even if not stated.
We assume algorithms used as subroutines in composable coresets are in $\mathrm{P}$ (computable in polynomial time), even if it is not repeated in the definitions.

\subsection{Preliminaries}\label{sec:defs}\label{sec:pre}
First, we review the definition of a clustering problem called metric $k$-center, then, we discuss two data-distributed models called massively parallel computations (MPC) and composable coresets.
Both sequential approximation algorithms for metric $k$-center use proofs based on maximal independent set, disk graph, dominating set, and triangle inequality property of metric spaces which we also discuss.
Finally, we review some graph theory results related to a family of graphs called expander graphs with applications in designing MPC algorithms.
\subsubsection{Metric $k$-Center}
Given a set of $n$ points in a metric space with distance function $d$ and an integer $k$, the $k$-center~\cite{vazirani2013approximation} problem asks for a subset of $k$ input points ($S$) called centers such that the maximum of the distances from an input point to its nearest center, denoted by $r$ is minimized. Formally,
\[
k\text{-center}(S)=\min_{C\subset S, |C|\leq k} \max_{p\in S} \min_{q\in C} d(p,q),
\]
where $d$ is a metric. Metric $k$-center has $2$-approximation algorithms with running times $O(nk)$~\cite{gonzalez} (Gonzalez's algorithm) and $O(n^2\log n)$~\cite{hochbaum1985best,vazirani2013approximation} (parametric pruning). Gonzalez's algorithm is a greedy algorithm that adds the farthest point as a center and repeats this process until $k$ centers are found. Parametric pruning tries all pairwise distances between points as candidates for the optimal radius $(r)$ in increasing order by putting disks of radius $r$ on the input points and removing the points covered by these disks until all points are covered; The algorithm fails if more than $k$ points are chosen as the centers.
A tight example for the approximation factor $2$ is given in~\cite{vazirani2013approximation}.

\subsubsection{Disk Graph}
For a set of shapes in the plane, the {\em intersection graph} has a vertex for each shape and an edge between a pair of vertices whose corresponding shapes intersect.
The {\em disk graph} of radius $r$ is the intersection graph of disks of radius $r$ centered at input points. Let $H(r,S)$ denote the disk graph of radius $r$ on a point set $S$. In metric spaces, this is the graph with $S$ as its vertices and an edge between a pair of points if their distance is at most $r$.

\subsubsection{Composable Coresets}
A set of coresets computed independently whose union gives an approximation of some measure $f$ of the whole point set (the union of the points in the sets) that is to be optimized is called a composable coreset~\cite{indyk2014composable}. Formally, for a set of sets $S_1,S_2,\ldots,S_L$, and a coreset construction algorithm $g$, assuming the problem is a minimization problem, an $\alpha$-approximation composable coreset satisfies
\[
f(S_1\cup \cdots \cup S_L) \leq f(g(S_1) \cup \cdots \cup g(S_L)) \leq \alpha  f(S_1\cup \cdots \cup S_L).
\]
So, a composable coreset is defined with the pair $(f,g)$.

\subsubsection{Massively Parallel Computation (MPC)}
A model with distributed data called the massively parallel computation (MPC) model~\cite{beame2013communication}. In MPC the space (memory) restrictions are as follows: each set (machine) has a sublinear size, the number of machines is sublinear, and the total memory is linear. The time restrictions of one round of MPC computation are that each machine is allowed to perform a polynomial-time computation independently from other machines and at the end of the round, the machines can send one-way communications to each other. We discuss MPC algorithms with a constant number of rounds.

Composable coresets of sublinear size with sublinear-size sets $S_i$, for $i=1,2,\ldots,L$ and coreset functions $g$ that take linear space and polynomial time are in MPC using a single-round MPC algorithm that computes $g(S_i)$ in each machine in parallel and then sends the results to one of the machines\footnote{The MapReduce implementation takes two rounds: one round to compute the coresets and change the keys of their results to machine $1$ and another one to send the mapped points from the machines to the first machine.}.

\section{A MPC Algorithm for $k$-Center}
First, we show that a composable coreset using any $2$-approximation algorithm for $k$-center gives a composable coreset with approximation factor $2$. Then, we show that if a modified version of the parametric pruning algorithm for $k$-center is used it would be possible to extract the same solution from the resulting composable coreset in all machines.

A covering of a point set with a set $C$ of balls of radius $r$ means each point of $S$ falls inside at least one ball.
We use a more restricted definition of covering in metric spaces in \Cref{def:cover}.
\begin{definition}[Center-Covered]\label{def:cover}
In a metric space, if a set $C$ with $\size{C}=k$ for a radius $r$ has the property that all the points of $S$ are within distance at most $r$ of a point $C$, we say that $C$ center covers $S$ with radius $r$. 
\end{definition}

In \Cref{lemma:2r}, we show that in a metric space, it is enough to keep one point from each cluster in a $k$-center solution to achieve a $2$-approximation of the cost of that solution. 
\begin{lemma}\label{lemma:2r}
In a metric space, if a set $C$ has one point from each cluster of a $k$-center solution of radius $r$, then, $C$ contains a $k$-center solution with radius at most $2r$. 
\end{lemma}
\begin{proof}
Let $M$ be the set of points from each of the clusters.
Based on the definition of $k$-center and center cover (\Cref{def:cover}), a solution $C'$ to $k$-center with radius $r$ is a center cover of radius $r$, so, $\forall m\in M, c\in C, \dist(c,m) \leq r \Rightarrow \forall p: \dist(c,p) \leq r, \dist(m,p) \leq \dist(m,c)+\dist(c,p) \leq 2r.$. The last inequality is due to the triangle inequality property of triples of points in a metric space.
\end{proof}

\begin{lemma}\label{lemma:cover}
A dominating set $D$ of the disk graph $H(r,S)$ center-covers the points of $S$ with radius $r$.
\end{lemma}
\begin{proof}
Since $D$ is a dominating set of $H(r,S)$, based on the definition of dominating sets, any point is either in $D$ (center covered with radius $0$) or adjacent to a vertex in $D$ (center covered with radius $r$). So, $D$ center covers $S$ with radius $\max(0,r)=r$.
\end{proof}

\begin{definition}[Permutation-Stable Algorithm]\label{def:perm}
Consider an algorithm $A$ that takes a permutation $\phi$ in addition to its input and returns the same solution if $\phi$ is not changed.
\end{definition}

The ordering of the points can be implemented by adding an extra dimension that is ignored when computing the distances.
\begin{algorithm}[h]
\begin{algorithmic}[1]
\Require{A point set $S$, an order $\phi$ on S, an integer $k$}
\Ensure{A set of pairs $(\rho,C)$ where $C$ center-covers $S$ with radius $\rho$}
\State{$W\gets\emptyset$}
\State{$\kappa\gets k$}
\State{Index the elements of $S$ in the order of $\phi$ and let $s_i=s\in S:\phi(s)=s_i$.}
\State{$R\gets\{\dist(p,q)\mid p,q\in S\}$}
\State{Sort $R$ increasingly.}
\While{$R\neq \emptyset$}
\State{$\rho\gets \min_{r\in R} r$}
\State{$C\gets \emptyset$}
\For{$i=1,\ldots,n$}
\If{$\size{C}<k$ and $s_i$ is not center-covered by $C$ with radius $\rho$}
\State{$C\gets C\cup \{s_i\}$}
\EndIf 
\EndFor
\If{$\size{C}= \kappa$}
\State{$W\gets W\cup (\rho,C)$}
\State{$\kappa\gets \kappa-1$}
\EndIf
\State{$R\gets R\setminus \{\rho\}$}
\EndWhile
\end{algorithmic}
\caption{Permutation-Stable Parametric Pruning for $k$-Center}\label{alg:param}
\end{algorithm}
\Cref{alg:param} takes $O(n)$ time to index the points, $O(n^2)$ time to compute the pairwise distances to build the set $R$, $O(n^2 \log n)$ time to sort them, and $O(n^3)$ time in the \textbf{while} loop to check if $k$ centers exist that center-cover the input points, for each element of $R$.

\begin{lemma}\label{lemma:param}
\Cref{alg:param} finds the first maximal independent set of the disk graph $H(r,S)$ with $r\leq 2r*$  in the order of $\phi$, where $r$ is the radius $(\rho)$ for which the algorithm terminates.
\end{lemma}
\begin{proof}
The algorithm visits the vertices in the order of $\phi$, so, by construction, it finds the first maximal independent set.
In the parametric pruning algorithm for $k$-center, the order is arbitrary, so, using order $\phi$ is allowed. The proof of the approximation ratio of parametric pruning shows the solution of the algorithm is a maximal independent set of the disk graph of radius $r$.
\end{proof}

\begin{definition}[Ordered Composable Coresets]\label{def:ordered}
Consider a composable coreset $T$ on a set of sets $S_1,S_2,\ldots,S_L$ for function $f:S\rightarrow \mathbb{R}^{\geq 0}$ using a permutation-stable coreset algorithm $g$ and an ordering defined on the input points $\phi:\cup_{i=1}^L S_i\rightarrow \{1,2,\ldots, \size{S}\}$. We define an ordered composable coreset as the tuple $C$ of the elements of set $T$ sorted in the order of increasing $\phi$.
An ordered composable coreset is defined as the tuple $(f,g,\phi)$.
\end{definition}

In \Cref{lemma:recovery}, we show that a $k$-center solution can be extracted from an ordered composable coreset if the coreset is a $k$-center solution.
\begin{lemma}\label{lemma:recovery}
For any ordered composable coreset $C$ for $k$-center where $g$ computes a solution for $k$-center as the coreset, there is a set $X\subset C$ of size $k$ and a radius $r\geq 0$ that center covers the points of $S_i$, for all $i=1,\ldots,L$.
\end{lemma}

\begin{proof}
Based on the assumption in the statement of the lemma and the definition of the coreset construction algorithm in the definition of ordered composable coresets, $g_{\phi}$ is a permutation-stable algorithm for $k$-center. So, when $g_{\phi}$ is applied to sets $S_i\cup C$, it finds a subset of the centers in $X$:
\[
g_{\phi}(S_i\cup X)\subset X.
\]
In an ordered composable coreset for $k$-center, if all the solutions $g_{\phi}(S_i)$ are subsets of a solution $X$ of size $k$, then the size of the union of the coresets is at most $k$:
\[
\cup_{i=1}^L g_{\phi}(S_i) \subset X \Rightarrow \size{\cup_{i=1}^L g_{\phi}(S_i) } \leq \size{X}=k.
\]
\qed
\end{proof}

We claim that \Cref{alg:exist} finds a $2$-approximation solution for $k$-center in MPC using an ordered composable coreset based on the parametric pruning algorithm for $k$-center. From \Cref{line:find} to the end of the algorithm, the algorithm finds the threshold for the prefixes of the set of centers that have to be chosen and the radius. From \Cref{line:second} to the end of the algorithm runs locally in the first machine.

\begin{algorithm}
\caption{A $2$-Approximation for Metric $k$-Center in MPC}\label{alg:exist}
\begin{algorithmic}[1]
\Require{Sets $S_i$, for $i=1,\ldots,L$, integer $k$}
\Ensure{A set of centers $T$ and a radius $\rho$}
\State{Assume an ordering $\phi$ on the points.}\label{line:1}
\State{Run \Cref{alg:param} on each $S_i$, and let $r_i$ be the cost and $C_i$ be the set of centers for $\size{C_i}=k$.}
\State{Let $C=\cup_{i=1}^L C_i$.}
\State{Send $C$ to all machines.}
\State{Update $\phi$ to prioritize the points of $C$ to $S\setminus C$ and keep the order of points in $C$.}\label{line:2}
\For{$\kappa=1,2,\ldots,k$ \textbf{in parallel}}
\State{Run \Cref{alg:param} on each set $S_i\cup C$ in the order of $\phi$ and let $T_{i,\kappa}$ be its set of centers and $\rho_{i,\kappa}$ be its radius.}\label{line:second}
\EndFor
\State{Send $T_{i,\kappa}$ and $\rho_{i,\kappa}$ to the first machine.}
\State{$R\gets \cup_{i=1}^L \cup_{\kappa=1}^{k} \{\rho_{i,\kappa}\}$}\label{line:find}
\State{Sort $R$ increasingly.}
\For{$\rho\in R$ in increasing order}
\For{$i=1,\ldots,L$ \textbf{in parallel}}
\State{$t_i\gets \max_{\substack{\kappa=1,\ldots,k,\\ r_{i,\kappa}\leq \rho}} \kappa$ }
\EndFor
\If{$\size{\cup_{i=1}^L T_{i,t_i}}\leq k$}
\State{\textbf{return} $T=\cup_{i=1}^L T_{i,t_i}$ and $\rho$}
\EndIf
\EndFor
\end{algorithmic}
\end{algorithm}
\Cref{alg:exist} takes $3$ rounds in MPC: the beginning of the rounds are \Cref{line:1,line:2,line:find}. The size of the coreset computed in \Cref{line:second} is $\size{\cup_{i=1}^L \cup_{\kappa=1}^k T_{i,\kappa}}=\sum_{i=1}^L \sum_{\kappa=1}^{L} \kappa=O(Lk^2)$ which dominates the rest of the communications used in the algorithm.

\begin{lemma}\label{lemma:zilla}
The set $C$ in \Cref{alg:exist} contains a $2$-approximation coreset for the metric $k$-center of $S$.
\end{lemma}
\begin{proof}
For each cluster $O_j$ of the optimal $k$-center on $S=\cup_{i=1}^L S_i$ with radius $r^*$ and center $o$, we prove one of these cases hold for the set $C$ in \Cref{alg:exist}:
\begin{enumerate}
\item the center $o_j$ of $O_j$ is in $C_i$, i.e., $\exists o_j\in C_i\cap O_j: \forall o\in O_j, \dist(o,o_j)\leq r^*$,
\item one point from $O_j$ is chosen in $C_i$, i.e., $\size{O_j\cap C_i}=1$,
\item the set $O_j\cap S_i$ is covered by a set of points such as $c\in C_i$ with radius $2r^*$,
\item there is only one point of $O_j$ in $S_i$, i.e., $\size{O_j\cap S_i}=1$, or
\item the cluster $O_j$ is missing from $S_i$, i.e., $O_j\cap S_i=\emptyset$.
\end{enumerate}
\Cref{fig:cases} shows these cases.
\begin{figure}[h]
\begin{subfigure}{\textwidth}
\centering
\includegraphics[scale=0.7]{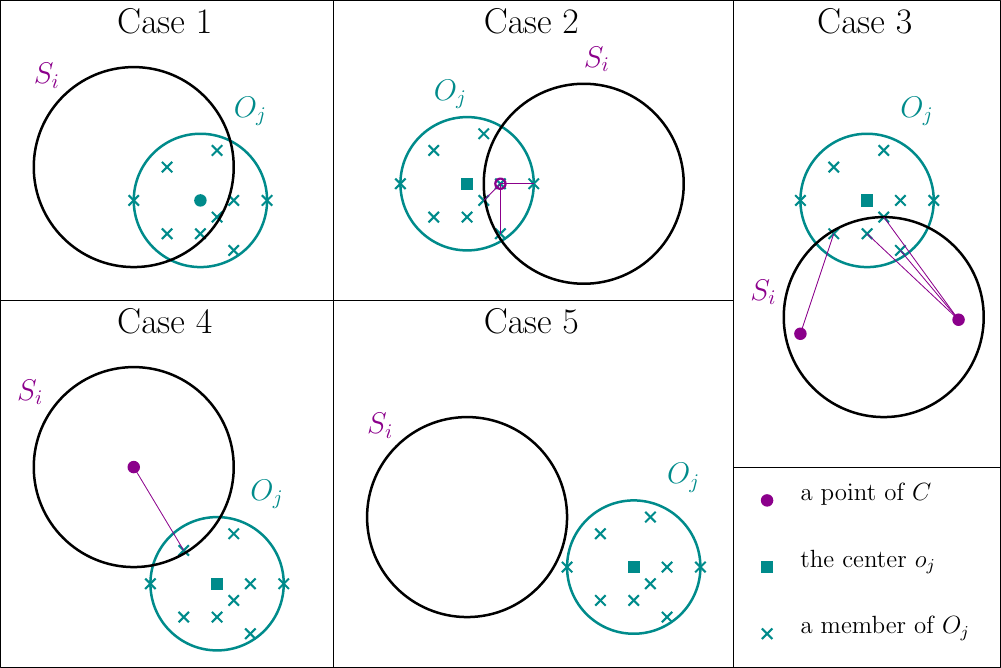}
\end{subfigure}
\begin{subfigure}{\textwidth}
\centering
\includegraphics[scale=0.75]{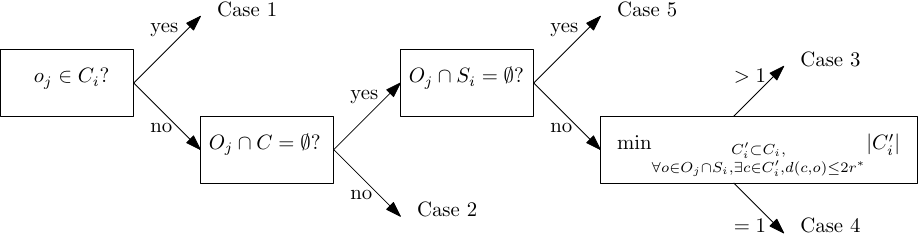}
\end{subfigure}
\caption{All the cases of covering $O_j \cap S_i$ for each set $S_i$ in the input.}\label{fig:cases}
\end{figure}
Assuming the maximum radius of the clustering using $C_i$, for $i=1,2,\ldots,L$ in all of the sets is at most $2r^*$, then this is all the cases because if $O_j\cap S_i\neq \emptyset$, there must be at least one center in $C_i$ that covers $O_j\cap S_i$. It remains to show that the radius of $C$ is at most $2r^*$ and a subset of at most $k$ points of $C$ cover $S=\cup_{i=1}^L S_i$. If we show this for the intersection of each cluster of the optimal solution and each input set, it always holds by induction on $L$. The base case is $L=1$, where there is only one set, which is the same as the serial algorithm.
We use another induction on the number of clusters $(k)$. The base case $(k=1)$ holds, as choosing any point from the input gives a $2$-approximation.

We use two facts in this proof:
\begin{itemize}
\item
When the number of clusters of the optimal solution in a set $S_i$ is less than $k$ and the algorithm chooses $k$ centers, the radius either decreases or remains the same.
\item
At twice the radius of an optimal clustering $(2r^*)$, all the members of each cluster form a clique. So, the subset of points of each optimal cluster in each set $S_i$ form a clique, so, the maximal independent set of the graph has at most as many points as there are clusters of the optimal solution in $S_i$, which is at most $k$. \Cref{lemma:cover} proves any dominating set, including the maximal independent set of a disk graph of radius $2r^*$ center-covers the points with radius $2r^*$.
\end{itemize}

Case 1: Since $o_j \in C_i$,
\[
\forall o\in (O_j\cap S_i), \dist(o,o_j)\leq r^*.
\]

Case 2: There is a point $c\in (O_j\cap C_i)$, so, using the triangle inequality on the points $o\in (O_j\cap C_i)$, $c$ and $o_j$, we have:
\[
\dist(c,o)\leq \dist(c,o_j)+\dist(o,o_j) \leq 2r^*.
\]

Case 3: Since there are no points of cluster $O_j$ in $C_i$, based on the pigeonhole principle, there are two points of the same optimal cluster, so, using the triangle inequality $r_i\leq 2r^*$. Centers in $C_i$ cover all the points of $S_i$, including $O_j\cap S_i$, with radius at most $2r^*$. Also, this case never happens for radius $2r^*$ and more as the set of vertices $C_i$ would not form an independent set which is the output of \Cref{alg:param}. So, in this case, the algorithm would find at most $k-1$ centers.

Case 4: The cluster $O_j$ has only one point $o$ in $S_i$ and $o\ni C_i$. So, at least two points from one of the clusters have been chosen, which gives the radius at most $2r^*$. In this case, the number of centers also decreases as the algorithm would not choose two points from the same optimal cluster for $2r^*$.

Case 5: The set $S_i$ contains none of the points in $O_j$, so, the number of clusters in $S_i$ is $k-1$.

Based on the induction on $k$, the radius $2r^*$ is enough to cover the input using $kL$ centers.
The induction on $L$ completes the proof.
\end{proof}

\begin{theorem}
\Cref{alg:exist} computes a $2$-approximation for $k$-center.
\end{theorem}
\begin{proof}
\Cref{lemma:zilla} proved that $C$ center-covers $S$ with radius $2r^*$ using $kL$ centers.
Case 3 in \Cref{lemma:zilla} is the only case where the number of centers can increase because more than one point is used to cover a cluster, which does not happen for radius $2r^*$.
After $C$ is added to the sets, since it is a permutation-stable algorithm as proved in \Cref{lemma:param}, the points of $C$ are used as centers, so, at most one point from each optimal cluster is chosen: the point that appears first in the order of $\phi$.
Based on \Cref{lemma:recovery}, the prefix of size $k$ of the solution (the union of the solutions in the sets) is that $2$-approximation solution.  
So, in all cases of \Cref{lemma:zilla} that happen for $2r^*$, either one point or no point is chosen.
This means the solutions of $\kappa$-center for $\kappa=1,2,\ldots,k$ whose union gives $k$ centers also have radius at most $2r^*$.
\end{proof}

The communication complexity of \Cref{alg:exist} is $O(k^2L)$ since the total size of coresets that build $C$ is $\sum_{i=1}^L \size{C_i}\leq kL$ and the total size of the coresets that build $T$ is $\sum_{i=1}^L \sum_{\kappa=1}^k \size{T_i}\leq Lk^2$.
The round complexity of \Cref{alg:exist} is $O(1)$: one round to compute the sets $C_i$, for $i=1,2,\ldots,L$, one round to broadcast $C$, one round to compute sets $T_{i,\kappa}$ for $i=1,2,\ldots,L$, $\kappa=1,2,\ldots,k$, and another round to locally compute $T$.
Also, \Cref{alg:exist} is in MPC only if $k^2L=O(m)$.

\section{Conclusion and Open Problems}
We gave a $2$-approximation algorithm in MPC for metric $k$-center by first designing a sequential approximation algorithm that finds a relaxed version of the lexicographically first solution (the first solution in some ordering which can be easier to compute than the lexicographically first order) from the set of possible $2$-approximation solutions for that problem and then using a superset of the points that contained that solution to locally recover it. This resolved the problem of choosing parts of different approximate solutions in different machines, leading to an improved approximation ratio.

\bibliographystyle{unsrt}
\bibliography{refs.bib}

\end{document}